\documentclass[a4paper,accepted=2022-08-17]{quantumarticle} 

\pdfoutput=1 
\usepackage{amsmath,amsthm,amsfonts,amssymb,amscd,amsbsy}
\usepackage{graphicx}
\usepackage{array}
\usepackage[numbers]{natbib}
\usepackage{enumerate}
\usepackage[nodayofweek]{datetime}
\usepackage[english]{babel}
\usepackage{mathtools}
\usepackage[toc,page]{appendix}
\usepackage{verbatim} 
\usepackage{amsthm} 
\usepackage{enumitem}
\usepackage{enumerate}
\usepackage{bbm} 
\usepackage[normalem]{ulem} 
\usepackage{bm}
\usepackage{authblk}

\usepackage[utf8]{inputenc}
\usepackage[pdfpagelabels]{hyperref}

\usepackage{mathrsfs}
\usepackage{amsmath}
\usepackage{amssymb}
\usepackage{amsfonts}
\usepackage{amsthm,bbm}
\usepackage{graphicx}
\usepackage[margin=3cm,]{geometry}

\makeatletter
\@addtoreset{equation}{myequation}
\makeatother

\usepackage{hyperref}
\hypersetup{colorlinks=true,allcolors=[rgb]{0,0,0.6}}

\newtheorem{thm}{Theorem}

\newtheorem{cor}[thm]{Corollary}

\theoremstyle{definition}

\newtheorem{rmk}[thm]{Remark}

\newcommand{\be}{\begin}
\newcommand{\e}{\end}
\newcommand{\beq}{\begin{equation}}
\newcommand{\eeq}{\end{equation}}

\numberwithin{equation}{section}
\numberwithin{thm}{section}

\newcommand{\Z}{{\mathbb Z}}

\newcommand{\C}{{\mathbb C}}

\newcommand{\de}{{\delta}}
\newcommand{\eps}{{\varepsilon}}

\newcommand{\Lam}{{\Lambda}}
\newcommand{\gam}{{\gamma}}

\newcommand{\De}{\Delta}

\renewcommand{\l}{\left}
\renewcommand{\r}{\right}

\newcommand{\Hm}[1]{\leavevmode{\marginpar{\tiny%
$\hbox to 0mm{\hspace*{-0.5mm}$\leftarrow$\hss}%
\vcenter{\vrule depth 0.1mm height 0.1mm width \the\marginparwidth}%
\hbox to 0mm{\hss$\rightarrow$\hspace*{-0.5mm}}$\\\relax\raggedright
#1}}}

 \usepackage{amssymb}
 \usepackage{comment}
 \usepackage{graphics}
\usepackage{bbm}

\begin{document}

 \title{Random translation-invariant Hamiltonians and their spectral gaps} 
 \author{Ian Jauslin}
 \affiliation{Department of Mathematics, Rutgers University, Piscataway, NJ 	08854, USA}
 \author{Marius Lemm}
 \affiliation{Department of Mathematics, University of Tübingen,  72076 Tübingen, Germany}

  \date{August 21, 2022}

 \begin{abstract}
We consider random translation-invariant frustration-free quantum spin Hamiltonians on $\Z^D$ in which the nearest-neighbor interaction in every direction is randomly sampled and then distributed across the lattice. Our main result is that, under a small rank constraint, the Hamiltonians are automatically frustration-free and they are gapped with a positive probability. This extends previous results on 1D spin chains to all dimensions. The argument additionally controls the local gap. As an application, we obtain a 2D area law for a cut-dependent ground state via recent AGSP methods of Anshu-Arad-Gosset.
\end{abstract}

	 	 \maketitle

\section{Introduction}
In quantum many-body physics, the existence of a spectral gap above the ground state has massive consequences for ground state correlation and entanglement properties \cite{arad2017rigorous,hastings2007area,hastings2006spectral,nachtergaele2006lieb}. The closing of a spectral gap is also intimately connected to the occurrence of topological quantum phase transitions, since the modern definition of a quantum phases relies on the existence of a path of gapped Hamiltonians through Hastings' notion of quasi-adiabatic evolution \cite{bachmann2012automorphic,hastings2004lieb,nachtergaele2019quasi}. The stability of spectral gaps under various ``local'' perturbations of the Hamiltonian is an active research area \cite{bravyi2010topological,del2021lie,de2019persistence,michalakis2013stability,nachtergaele2021stability}, and to make use of these stability results it is of course beneficial to have a wide net of gapped Hamiltonians available for further stability analysis. Generally, questions concerning spectral gaps are at the heart of many of the most challenging open problems in physics. Two examples are Haldane's conjecture that the antiferromagnetic Heisenberg chain is gapped for integer spin values \cite{haldane1983continuum,haldane1983nonlinear}. and the Yang-Mills mass gap, a millennium problem. For additional background about the relevance of spectral gaps, see \cite{lemm2019spectral,nachtergaele2019quasi}.

Given the fact that the existence of a spectral gap has strong physical implications, there has been significant interest in identifying mathematical techniques for deriving a spectral gap rigorously. It has been found that with very few exceptions, only special frustration-free Hamiltonians are amenable to a rigorous analysis of the spectral gap. We recall that frustration-free Hamiltonians play a central role in quantum many-body theory, for example for classification of topological phases of quantum matter and for quantum error correcting codes. For a highly incomplete list of examples, we mention that the toric code and Levin--Wen models are paradigms of topological quantum computation \cite{kitaev2006anyons,levin2005string}, AKLT ground states \cite{affleck1988valence} are universal resource states for measurement-based quantum computation  \cite{miyake2011quantum,verstraete2004valence,wei2011affleck,wei2014hybrid}, Motzkin Hamiltonians provided the first examples with volume-entanglement \cite{movassagh2016supercritical,zhang2017novel}, and product vacua with boundary states (PVBS) \cite{bachmann2014product,bachmann2015product,bishop2016spectral} provided hands-on  models for topological quantum phase transitions. Frustration-free Hamiltonians also arise generally as the parent Hamiltonians of matrix product states (MPS) or projected entangled pair states (PEPS) and are thus central to the study of tensor network states. The main reason for their usefulness is that they are among the few many-body Hamiltonians that can be rigorously analyzed. 

The seminal work of Affleck-Kennedy-Lieb-Tasaki \cite{affleck1988valence} derived the spectral gap for their eponymous (frustration-free) AKLT chain. This has led to a large number of follow-up works which also focus on special frustration-free Hamiltonians which can be analyzed by one of the standard methods, the martingale method \cite{bachmann2014product,bachmann2015product,bishop2016spectral,nachtergaele1996spectral} or finite-size criteria \cite{abdul2020class,lemm2019gapped,lemm2020existence,lemm2019aklt,pomata2019aklt,pomata2020demonstrating}, or their combination, \cite{nachtergaele2021spectral,warzel2022bulk}.

While the study of these specific frustration-free Hamiltonians leads to numerous insights, it leaves open the physically very relevant question whether generic frustration-free Hamiltonians are gapped or not. Part of the motivation for asking this question is that random Hamiltonians are considered key models for quantum chaos. As a consequence, a new avenue of research \cite{bravyi2015gapped,lemm2019gaplessness,movassagh2017generic} has recently emerged which asks the question whether gaps are common or uncommon.  

It is worth pointing out that there exist heuristic physical arguments both for and against the commonality of gaps: On the one hand, as reviewed above, gaps place extremely strong constraints on ground states which suggests they are rare in some sense. On the other hand, the closing of the gap is associated with a quantum phase transition, which suggests that most systems should be gapped (i.e., not at the brink of a quantum phase transition). To further highlight that this is a subtle problem, we recall that Movassagh \cite{movassagh2017generic} proved that spatially disordered Hamiltonians are generically gapless in any dimension. Movassagh's argument fundamentally relied on the decoupling of rare local regions from the bulk of the system and thus fundamentally on the non-translation-invariance of the Hamiltonians. A natural follow-up question was then what happens in the translation-invariant setting where this mechanism is absent. Concretely, we would like to address the following  question:\\
 
\textit{Does a typical translation-invariant frustration-free local Hamiltonian have a spectral gap?}\\
 
Here the word ``typical'' refers to an underlying probability measure on the space of translation-invariant, nearest-neighbor Hamiltonians which is made precise in Subsection \ref{ssect:model} below. The main model parameters are
$$
\begin{aligned}
D=&\textnormal{spatial lattice dimension},\\
 d=&\textnormal{local qudit dimension},\\ 
 r=& \textnormal{ rank of interaction}.
\end{aligned}
$$

The Hamiltonians are normalized to have local ground state energy zero and so the rank $r$ equals $d^2-\textnormal{local ground state dimension}$. We will be interested in models of sufficiently small rank $r$, or equivalently, of sufficiently large \textit{local ground space dimension}. The point is that choosing a sufficiently small $r$ is a way to ensure frustration-freeness which is well-known for line graphs \cite{movassagh2010unfrustrated}.

Indeed, this makes intuitive sense: A small rank is the same as a large local ground space dimension. This makes it easier for the local ground states to overlap sufficiently to obtain a global ground state which does not frustrate any local interaction. As a consequence, the well-known examples frustration-free Hamiltonians have only moderate rank $r$ compared to the maximal allowed rank $d^2$ for nearest-neighbor interactions. For example, the AKLT chain has $r=5$ and $d^2=9$, Motzkin Hamiltonians have $r=3$ and $d^2=9$, and product vacua with boundary states (PVBS) have $r=(d-1)(d+2)/2$. This heuristic beckons the question whether small $r$ is sufficient to ensure frustration-freeness by a general argument. Here, we will give a general, rigorous proof that the resulting Hamiltonians are always frustration-free if $r$ is sufficiently small. See Lemma \ref{lm:FF} for the precise mathematical statement. This relies on a QSAT criterion from \cite{sattath2016local} and the rigorous Kotecky-Preiss criterion for the convergence of the cluster expansion from classical statistical mechanics \cite{kotecky1986cluster}.  

Regarding the experimental realizability of the small rank constraint, we first recall that a standard setup to realize quantum spin systems are cold quantum gases, e.g., constrained by optical lattices. A common way to increase single-bond ground state degeneracy (i.e., to lower the interaction rank) is to coarse-grain the Heisenberg antiferromagnetic two-spin interaction by introducing the higher-order coupling $(\vec{S}_x\cdot \vec{S}_y)^2$ to obtain a spectral projector as is done, e.g., in the AKLT Hamiltonian. Orbital-dependent Hubbard models are expected to display these higher-order couplings and hence e.g.\ the AKLT Hamiltonian experimentally \cite{koch2015affleck}. While it can be challenging to realize small interaction ranks in the laboratory, this challenge is shared among almost all frustration-free Hamiltonians. In summary, progress on the ongoing program of realizing frustration-free Hamiltonians in the laboratory will likely translate to the Hamiltonians considered here. A particular advantage of the construction we provide is that it is flexible about the underlying graph structure and can thus be readily adapted to pertinent experimental constraints.

With Lemma \ref{lm:FF} at hand, we have produced a natural model of translation-invariant frustration-free nearest-neighbor Hamiltonians. The next step is to analyze the typical spectral gap. Our main result says that these random Hamiltonians are gapped with positive probability independent of system size (Theorem \ref{thm:main}). Together with \cite{movassagh2017generic}, which treated the spatially disordered case, this proves rigorously that gaps are strictly more common in the presence of translation-invariance. 

As reviewed below, our first result is the first that applies to graphs of dimension $>1$. Previous works \cite{bravyi2015gapped,lemm2019gaplessness} in the frustration-free setting were restricted to quantum spin chains, i.e., line graphs.\\ 

Our result paves the way to a broader understanding of the gaps of frustration-free Hamiltonians and their ground state properties. Given that these ground states and their gaps are central tools for the modern definition of quantum phases, understanding their landscape is of clear theoretical relevance. In addition, gaps have significant practical implications for the lifetimes of quantum states, e.g., for quantum memory or measurement-based quantum computation. For example, AKLT ground states on various graphs are universal resource states for measurement-based quantum computation \cite{miyake2011quantum,verstraete2004valence,wei2011affleck,wei2014hybrid} which are known to be protected by a gap since recently \cite{abdul2020class,lemm2020existence,pomata2019aklt,pomata2020demonstrating}. Our result raises the possibility of widely expanding the class of gap-protected universal resource states for MBQC by analyzing the ground states of the gapped Hamiltonians we identify here. We would like to mention that the method we use to prove the main result Theorem \ref{thm:main} not only yields positive probability of the spectral gap, it also identifies the types of local interactions that produce a gap. This is important for potentially realizing the Hamiltonians we find here physically. The details can be found in Remark \ref{rmk:main} (i).

An advantage of our method (finite-size criteria) is that it also gives a system size-independent lower bound on the \textit{local gap}. The local gap has recently received attention because it is needed to obtain the 2D area law in the recent breakthrough of Anshu-Arad-Gosset \cite{anshu2022area}. The problem in general is that the local gap can in principle 
be much smaller than the global gap. Here, we are able to prove that this does not happen by deriving a lower bound on the local gap of our random Hamiltonians. Consequently, we are able to provide for the first time a broad class of examples of 2D Hamiltonians satisfying the area law. See Subsection \ref{ssect:ALapplication} for the details.

\section{Model an main results}

\subsection{The random Hamiltonians} 
\label{ssect:model}
We consider a nearest-neighbor Hamiltonian acting on a system of qudits, described by local Hilbert spaces $\C^d$, which are placed on a $D$-dimensional discrete box $\Lam_L=((-L,L]\cap \Z)^D$ of side length $2L$. The total Hilbert space is $\bigotimes_{x\in \Lam_L} \C^d$ and the Hamiltonian is given by
\beq\label{eq:HLdefn}
H_{\Lam_L}=\sum_{x\in \Lam_L} \sum_{j=1}^D P^j_{x,x+e_j}
\eeq
with periodic or open boundary conditions. We say that the Hamiltonian is translation-invariant if all the local interactions pointing in the same direction act identically. Formally, this means we have
$$
\begin{aligned}
P^j_{x,x+e_j}=P^j\otimes \mathrm{Id}_{\Lam_L\setminus\{x,x+e_j\}},\\
 \textnormal{for all } x\in\Lam_L \textnormal{ and all } j=1,\ldots,D
\end{aligned}
$$
with fixed prototypes $P^1,\ldots,P^D:\C^{d}\otimes \C^{d}\to\C^{d}\otimes \C^{d}$, one for every direction; see Figure \ref{fig:grid}. Note that we do not assume isotropy of the local interaction, i.e., the prototypes $P^1,\ldots,P^D$ may differ for different directions. We assume that the local interaction is an orthogonal projection of rank $r$.

\begin{figure}[t]
\begin{center}
\includegraphics[width=75mm]{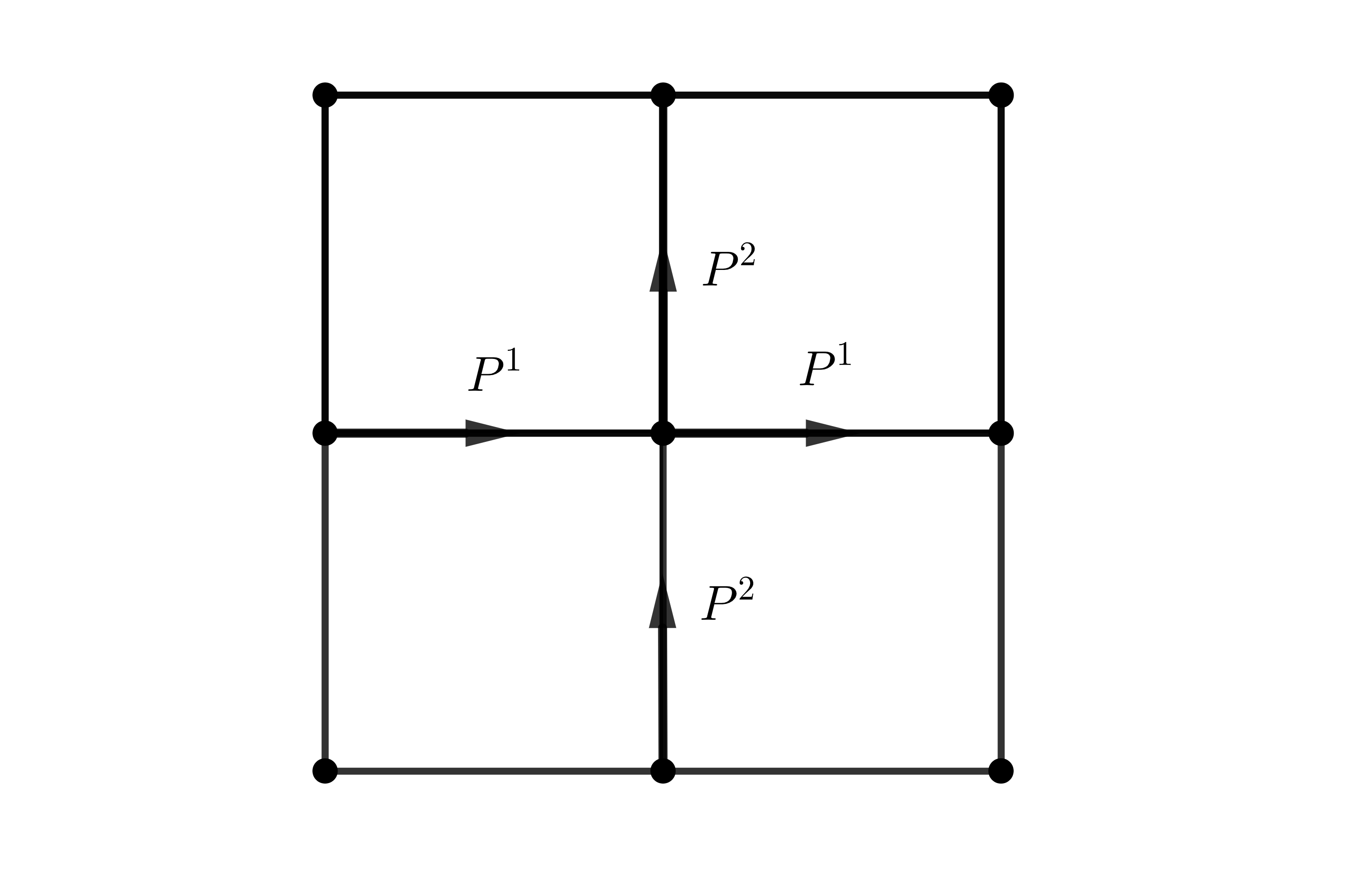}
\caption{The Hamiltonian \eqref{eq:HLdefn} on the square lattice: $P^1$ acts in the $e_1$-direction and $P^2$ acts in the $e_2$-direction.}
\label{fig:grid}
\end{center}
\end{figure}

To turn these $H_{\Lam_L}$'s into a family of \textit{random translation-invariant Hamiltonians}, we only need to sample the prototypes for local interactions $P^1,\ldots,P^D$ from a suitable probability distribution.  For simplicity, we assume henceforth that the random projections $P^1,\ldots,P^D$ are independent and identically distributed and each $P^j$ is the projection onto the first $r$ columns of a Haar-random orthogonal matrix. The use of Haar measure is natural; indeed, it is equivalent to selecting the $r$ spanning vectors uniformly from the unit sphere, with different vectors chosen independently modulo the orthogonality constraint. Many other choices are possible and easily implementable, cf. Remark \ref{rmk:main} (ii) below.

\begin{rmk}[Possible extensions]
We close the presentation with a brief discussion of the assumptions on the underlying probability measure. Recall that we took the distribution of the random projections $P^1,\ldots,P^D$ to be i.i.d.\ with  each $P^j$ given as projection onto the first $r$ columns of a Haar-random orthogonal matrix. While independence is helpful, some correlation between different directions can be accommodated, but the proof does not apply to the isotropic case $P^1=\ldots=P^D$. There is a lot of freedom in choosing the underlying probability measure of each $P^j$, as long as the small-rank assumption is verified. More precisely, the choice of probability measure only enters in the proof of \cite[Lemma 7]{lemm2019gaplessness} which is used to obtain \eqref{eq:lemma7} below. This lemma was proved in \cite{lemm2019gaplessness} for Haar measure but the proof extends verbatim to any probability measure such that basis vectors are independent modulo orthonormality and the induced probability measure for an individual basis vector assigns positive probability to each spherical cap. For interactions of full rank, one can instead use a perturbative result by Yarotsky \cite{yarotsky2005uniqueness} to obtain a gap with positive probability.
\end{rmk}


\subsection{Previous results}

So far, the literature on the spectral gap problem for random translation-invariant Hamiltonians has been restricted to quantum spin chains, meaning that $D=1$, so the graph is a line $\Lam_L=\{1,2,\ldots,L\}$. We denote the 1D Hamiltonian by 
$$
H_L^{\mathrm{1D}}\equiv H_{\Lam_L}=\sum_{x=1}^{L} P_{x,x+1}
$$
where we chose periodic boundary conditions, identifying $L+1\equiv 1$. The following is known in 1D.

\be{itemize} 
\item For local dimension $d=2$ and rank $r=1$, Bravyi-Gosset \cite{bravyi2015gapped} proved that, with probability $1$, the quantum spin chain Hamiltonian $H_L^{\mathrm{1D}}$ is gapped.
 \item For arbitrary local dimension $d$ and rank $r\leq d-1$, \cite{lemm2019gaplessness} proved that the quantum spin chain Hamiltonian $H_L^{\mathrm{1D}}$ is gapped with positive probability.
\e{itemize}

We recall that we say that the family of $H_L^{\mathrm{1D}}$'s is \textit{gapped} (or $H_L^{\mathrm{1D}}$ is gapped for short) if there exists a universal constant $c>0$ such that for all $L$, the spectral gap above the ground state satisfies $\gam(H_L^{\mathrm{1D}})\geq c>0$. We remark that \cite{bravyi2015gapped}  analyze the $d=2$ and $r=1$ case in additional detail; they obtain a complete description of the spectral gap phase diagram in terms of the one-dimensional subspace in $\C^4$ that one projects on.

\subsection{Main result}
The main result of this work breaks the dimensional barrier and extends the previous results of \cite{bravyi2015gapped,lemm2019gaplessness} from one-dimensional line graphs to arbitrary graphs for the first time. Similarly to those work, we need to assume that the rank $r$ is sufficiently small compared to the local qudit dimension $d$.

\be{thm}[Main result]\label{thm:main}
Assume that
\beq\label{eq:mainass}
r\leq \min\l\{
\frac{\lfloor\frac{d^2}{4}\rfloor}{D},\frac{e^{-1}}{4D-1} d^2\r\}.
\eeq
Then, $H_{\Lam_L}$ is gapped with probability bounded below by $p>0$. \e{thm}

More precisely, the result says that there exists $c>0$ and a positive probability $p>0$ such that
\beq\label{eq:cp}
\mathbb P(\gam(H_{\Lam_L})\geq c)\geq p>0.
\eeq
Importantly, the constants $c>0$ and $p>0$ depend only on the order-$1$ system parameters $D,d$ and $r$ and not on the system size $L$. Similarly, it is an important feature of condition \eqref{eq:mainass} that it only depends only on these order-$1$ system parameters and again $L$ is not involved. (For example, taking rank $r=1$ and spatial dimension $D=2$, condition \eqref{eq:mainass} holds whenever the local qudit dimension $d\geq 4$.) The fact that these quantities are independent of $L$ is crucial for a meaningful derivation of a spectral gap in the thermodynamic limit $L\to\infty$.

 In words, our result can be summarized as the following answer to (a more precise version of) the question formulated in the introduction:\\

\textit{A random translation-invariant Hamiltonians with local interaction of small rank has a positive probability of being gapped.}\\

We recall that the small-rank assumption should be understood as ensuring frsutration-freeness. The precise statement is in Lemma \ref{lm:FF}.

A few remarks about Theorem \ref{thm:main} are in order.

\begin{rmk}\label{rmk:main}
\be{enumerate}[label=(\roman*)]
\item As in \cite{lemm2019gaplessness}, the proof is based on the verification of a finite-size criterion with positive probability. The argument is constructive: it provides explicit choices of local projections $P$ for which the Hamiltonian is gapped. Concretely, the interaction $P$ is obtained by taking the local projections to be onto the span of $r$ of the vectors
$$
w_{ij}=\vert i\otimes j\rangle,\qquad \textnormal{for $i\leq d$ odd and $j\leq d$ even}.
$$
and then perturbing these vectors on the unit sphere of $\mathbb C^d\otimes \mathbb C^d$ within a sufficiently small spherical cap. Further details can be found in Subsection \ref{ssect:proofmain}. The fact that any sufficiently small perturbation inside the spherical cap yields a gap is practically relevant, as it ensures that the implementation is stable against perturbations of the local interaction which may be present in experiments.
\item  It is possible to make the constants $c$ and $p$ in \eqref{eq:cp} explicit. 

\item The proof generalizes to other graphs, e.g., the honeycomb lattice, as explained in Section \ref{sect:gen}.

\e{enumerate}
\e{rmk}

Let us also compare our result to undecidability results on the spectral gap problem \cite{bausch2020undecidability,cubitt2015undecidability}. Specifically, the work \cite{cubitt2015undecidability} constructs a 2D Hamiltonian whose spectral gap problem is equivalent to the halting problem for Turing machines and therefore algorithmically and logically undecidable. The relevant Hamiltonians require very special choices of local interactions which form a zero-measure set for almost any choice of probability measure on interactions. In view of this, it is very likely that the spectral gap problem is generically decidable. Hence, the question whether the spectral gap problem is decidable is somewhat orthogonal to the question about the typical value of the gap that we consider here. 

Despite these a priori different perspective, we believe that the usefulness of small rank that we leverage here may shed further light on the undecidability question. The point is that the special examples constructed in \cite{cubitt2015undecidability} require large rank. Our work then raises the question whether the small rank assumption can be leveraged even further to address decidability of the gap for sufficiently small rank interactions. Note that this is what Bravy-Gosset obtained in 1D, where there is a similar undecidability result \cite{bausch2020undecidability}.

\subsection{Small rank implies frustration-freeness}
A key step in the proof is to show that Assumption \eqref{eq:mainass} on the parameters ensures that $H_{\Lam_L}$ is frustration-free (i.e., $\inf\mathrm{spec}\, H_{\Lam_L}=\{0\}$). This is non-trivial because it leverages information about the \textit{local} model parameters $r$ and $d$ to derive the \textit{global} property of frustration-freeness.

The result holds for general graphs of bounded degree and may be of independent interest. Given a graph $G=(K,E)$, let 
$$
H_G=\sum_{e\in E} P_e
$$
be a nearest-neighbor Hamiltonian with all local interactions $P_e$ of rank $r$. We still denote the local qudit dimension by $d$.

\be{lm}\label{lm:FF}
Let $G$ be a graph with degree bounded by $\de$. Assume that
\beq\label{eq:FFass}
r\leq \frac{e^{-1}}{2\de-1} d^2.
\eeq
Then $H_G$ is frustration-free.
\e{lm}

An intuitive explanation for this result is again that small interaction rank corresponds to large local ground state space, which leaves more “room” in the Hilbert space for the system to find a global ground state that is also a local ground state everywhere (i.e., that is frustration-free).


This important lemma rests on verifying the statistical mechanics style QSAT criterion from \cite{sattath2016local} (recalled as Theorem \ref{thm:QSAT} below) through bounds on the partition function of a classical hard-core dimer system that we derive from the cluster expansion method (Koteck{\'y}-Preiss criterion \cite{kotecky1986cluster}).


 \subsection{Bounds on the local gap and application to the area law}\label{ssect:ALapplication}
In fact, our method lower bounds the local gap. We recall that the local gap plays a key role in the recent breakthrough of Anshu, Arad, and Gosset \cite{anshu2022area} who proved that non-degenerate ground states of locally gapped frustration-free 2D Hamiltonians satisfy an area law. This is a 2D analog of a famous 1D result by Hastings \cite{hastings2007area}. Here we will use a variant that is adapted to highly degenerate ground spaces \cite{AL}; see also \cite{abrahamsen2020sharp}.
 
 Generalizing the definition of local gap  from \cite[Eq.\ (41)]{anshu2022area} to arbitrary graphs, we obtain the following.
 
 \be{defn}[Local gap]
 Let $G$ be a graph. The local gap is given by
 \beq
 \gam_{\mathrm{loc},G} =\min_{S\subset G} \gam(H_S).
 \eeq
\e{defn}

 The new area law criterion from \cite{anshu2022area} requires a lower bound on $ \gam_{\mathrm{loc},G} $ at least when $S$ varies over rectangle in $\Z^2$. Consequently, lower bounds of the local gap have become of significant interest. Note the trivial bound \beq\label{eq:localgapvsgap}\gam_{\mathrm{loc},G} \leq \gam_G,
\eeq so a lower bound on the local gap is more difficult to derive than a lower bound on the usual, global gap. Fortunately, our argument can be strengthened to give a lower bound on the local gap.

\be{thm}[Local gap version of main result]\label{thm:mainlocal}
Under the assumptions of Theorem \ref{thm:main}, there exist $c>0$ and a positive probability $p>0$ such that
$$
\mathbb P(\gam_{\mathrm{loc},\Lam_L}\geq c)\geq p>0.
$$
\e{thm}

The only difference compared to Theorem \ref{thm:main} is the replacement of $\gam_{\Lam_L}$ by $\gam_{\mathrm{loc},\Lam_L}$. Our technique to control the local gap will be discussed in Subsection \ref{sect:fslocal}.

\begin{rmk}
While it is generally expected that the gap decreases with system size, which would suggest a reverse inequality to \eqref{eq:localgapvsgap}, this heuristic can likely fail for certain subsystems due to boundary effects of finite systems. It is an interesting open problem to derive any kind of converse to \eqref{eq:localgapvsgap}, even restricting to the local gap over rectangles that enters in \cite{anshu2022area}.
\end{rmk}

Recalling the area-law criterion from \cite{anshu2022area} and considering Lemma \ref{lm:FF} and Theorem \ref{thm:mainlocal}, we see that only the non-degenerateness of the ground space can fail for our models. In fact, one can use Theorem \ref{thm:QSAT} to see that the ground space is exponentially degenerate in system size for certain parameter values. While methods exist for extending the result of \cite{anshu2022area} to degenerate ground spaces \cite{abrahamsen2020sharp}, they come with an additional contribution to the entanglement entropy of $\log \dim(\textnormal{ground space})$ which becomes volume-like, and thus uninformative, for exponential ground space degeneracy. 

Instead, we rely on a modified version of the area law that is obtainable by combining the methods from \cite{anshu2022area} and \cite{arad2017rigorous}. In terms of techniques, these works rely on the approximate ground state projector (AGSP) that is well-suited for frustration-free Hamiltonians.

We now state this novel variant criterion \cite{AL}. Let the graph $G=([-L,L]\cap \Z)^2$ be a box in $\Z^2$. We split the Hilbert space into left and right halves in the natural way
$$
\mathcal H=\mathcal H_{\mathrm{left}}\otimes \mathcal H_{\mathrm{right}}.
$$
Given a ground state $\psi$ of $H_G$, we let
$$
\rho_{\psi}^{\mathrm{left}}=\mathrm{Tr}_{\mathcal H_{\mathrm{right}}}\vert\psi\rangle\langle\psi\vert
$$
equal the partial trace over the right half. The entanglement entropy of the left and right halves is then given by 
\beq
S(\rho_{\psi}^{\mathrm{left}})=-\mathrm{Tr}\rho_{\psi}^{\mathrm{left}}\log\rho_{\psi}^{\mathrm{left}}.
\eeq
The following result holds deterministically for any Hamiltonian of the form \eqref{eq:HLdefn}.

\be{thm}[Area law criterion, \cite{anshu2022area,AL}]\label{thm:ALcrit}
Let $D=2$. Suppose that $H_{\Lam_L}$ is frustration-free and there exists $c>0$ independent of $L$ such that $\gam_{\mathrm{loc},G}\geq c>0$. 

Then, there exists a ground state $\psi$ of $H_G$ that satisfies an area law for the entanglement entropy, i.e.,
\beq\label{eq:AL}
S(\rho_{\psi}^{\mathrm{left}})\leq CL^{1+(\log L)^{-1/5}}
\eeq
\e{thm}

As pointed out to us by \cite{AL}, this follows directly from the validity of the AGSP condition, cf.\ \cite{anshu2022area,arad2017rigorous}.

\begin{rmk}
The logarithmic correction in \eqref{eq:AL} does not play any important role. The result can be generalized if desired, e.g., to any cut and to rectangular (i.e., non-square) domains. Moreover, only a lower bound on the local gap over rectangles and L-shaped regions is needed as in \cite{anshu2022area}. 
\e{rmk}

Our main results then imply that the random translation-invariant Hamiltonians are among the small number of non-trivial examples of 2D Hamiltonians that provably have ground states satisfying an area law. 

\be{cor}[Area law]
Let $D=2$ and assume that \eqref{eq:mainass} holds.  

Then, with a positive probability $p>0$ that is independent of $L$, $H_{\Lam_L}$ has a ground state satisfying the area law \eqref{eq:AL}. 
\e{cor}

In closing, we remark that this result does not rely on PEPS techniques and, indeed, it does not identify the ground state satisfying the area law as a PEPS state.

\subsection{Finite-size criteria automatically control local gaps}\label{sect:fslocal}
In this short section, we would like to emphasize a methodological point. The local gap bound from Theorem \ref{thm:mainlocal} rests on the observation that the standard Knabe finite-size criterion for the gap \cite{knabe1988energy} can be easily extended to also lower bound the local gap, which as reviewed above has recently emerged as an important quantity in quantum information theory. This shows that there is an efficient method available for lower bounding the local gap---finite-size criteria. 

Let  $j,k=1,\ldots,D$. We introduce the local $3$-site Hamiltonians
\beq\label{eq:H3list}
\begin{aligned}
H_{j}=&P^j_{-e_j,0}+P^j_{0,e_j},\\
H^{++}_{jk}=&P^j_{0,e_j}+P^k_{0,e_k},\\
H^{-+}_{jk}=&P^j_{-e_j,0}+P^k_{0,e_k},\\
H^{--}_{jk}=&P^j_{-e_j,0}+P^k_{-e_k,0}.
\end{aligned}
\eeq
These local Hamiltonians cover all the possibilities by which adjacent edges can interact. We denote their common minimal spectral gap by
\beq\label{eq:gam3defn}
\gam_3=\min_{j,k=1,\ldots,D}\min\left\{\gam(H_{j}),\gam(H^{++}_{jk}),\gam(H^{-+}_{jk}),\gam(H^{--}_{jk})\r\}.
\eeq

\be{prop}[Finite-size criterion for the local gap in any dimension]\label{prop:knabelocal}
Let $H_{\Lam_L}$ be frustration-free. Then
\beq\label{eq:knabelocal}
\gam_{\mathrm{loc},\Lam_L}\geq (4D-2)\l(\gam_3-\frac{4D-3}{4D-2}\r) 
\eeq
\e{prop}

The proof is a generalization of the one from \cite{knabe1988energy}. It is included in the appendix for the convenience of the reader.  The number $\frac{4D-3}{4D-2}$ is commonly called the gap threshold. We refer to \cite{anshu2020improved,gosset2016local,lemm2020finite,lemm2019spectral,lemm2022quantitatively} for other improvements and generalizations of Knabe's finite-size criterion, especially for hypercubic lattices.

 \subsection{The one-dimensional case}
 Upon revisiting the construction in the one-dimensional spin chain setting \cite{lemm2019gaplessness}, we observe here that by tweaking the construction in \cite{lemm2019gaplessness}, the range of ranks can be improved up to a natural limit. The 1D random Hamiltonian is of the simple form
 $$
 H_L=\sum_{x=1}^{L-1} P_{x,x+1},
 $$
 where we chose open boundary conditions for definiteness. (The same proof applies verbatim to periodic boundary conditions.)
 
 \be{thm}[Improved 1D result]\label{thm:improvementL}
Let $r\leq \max\{\tfrac{d^2}{4},d-1\}$. Then $H_L$ is gapped with positive probability. 
 \e{thm}
 
 This strictly extends the main result of \cite{lemm2019gaplessness} whenever $\tfrac{d^2}{4}\geq d$, i.e., whenever $d\geq 4$. The assumption  $r\leq \max\{\tfrac{d^2}{4},d-1\}$ is precisely what guarantees frustration-freeness \cite{movassagh2010unfrustrated} and so the range of ranks covered by Theorem \ref{thm:improvementL} is a natural limit point of existing methods, as these all require frustration-freeness. 
 
 The restriction to $r\leq d-1$ in the previous work \cite{lemm2019gaplessness} arose due to a technical issue that we resolve here. See Appendix \ref{app:improvementL} for the details.

\section{Frustration-freeness}
In this section, we prove Lemma \ref{lm:FF} using the QSAT criterion Theorem \ref{thm:QSAT} from \cite{sattath2016local}. These results hold on general graphs.

The criterion concerns the classical partition function of the nearest-neighbor exlusion on the interaction graph associated to $H_G$. We take some care to formulate the criterion directly in a way that is suitable for our purposes:
\be{itemize}
\item We define the partition function with the notation of \cite{kotecky1986cluster}. This is convenient for implementing  the cluster expansion method to bound the partition function afterwards.
\item We focus on edge graphs which are the interaction graphs of nearest-neighbor interactions such as our $P_{x,x+e_j}$.
\e{itemize}

\be{defn}[Partition function of nearest-neighbor exclusion]\label{defn:Z}
Consider a graph $\mathcal G=(\mathcal K,\mathcal E)$. 
Given two vertices $\gamma_1,\gamma_2\in \mathcal K$, we write $\gamma_1\iota\gamma_2$ (read: $\gamma_1$ is joined to $\gamma_2$) or equivalently $(\gamma_1,\gamma_2)\in\iota$ if $\gamma_1$ and $\gamma_2$ are either equal or connected by an edge.
Let $\mathcal D_0(\mathcal G)$ be the set of \textit{admissible} subsets $\partial\subset \mathcal K$, i.e., for every $\gamma_1,\gamma_2\in\partial$, $(\gamma_1,\gamma_2)\not\in\iota$.
Let $\mathcal B(G)$ denote the set of subsets $\partial\subset K$ that are {\it connected}, in the sense that any partition $S_1,S_2\subset\partial$ with $S_1\cup S_2=\partial$ is such that there exists $\gamma_1\in S_1$ and $\gamma_2\in S_2$ such that $\gamma_1\iota\gamma_2$.
The grand-canonical partition function of the nearest-neighbor exclusion on $\mathcal G$ at fugacity $z\in\mathbb C$ is
\begin{equation}
  \mathcal Z(\mathcal G;z)=\sum_{\partial\in\mathcal D_0(\mathcal G)}z^{|\partial|},
\end{equation}
where $|\partial|$ is the cardinality of $\partial$.
\e{defn}

We will consider the nearest-neighbor exclusion not on the original graph $G$, but on its edge graph $G'$ defined as follows.

\be{defn}[Edge graph]
Given a graph $G=(K,E)$ we define the edge graph $G'$. The set of vertices of $G'$ is the set $E$ of edges of $G$. For every pair $e,e'\in E$, $G'$ has an edge between $e$ and $e'$ if and only if $e$ and $e'$ share a vertex in $G$ (that is, $e\cap e'\neq\emptyset$).
\e{defn}


The QSAT criterion reads as follows.

\be{thm}[QSAT criterion \cite{sattath2016local}]
\label{thm:QSAT}
Let $G$ be a graph. Let $p=\frac{r}{d^2}$ and suppose that 
\beq\label{eq:Zpos}
\mathcal Z(G';-p')>0,\qquad p'\in [0,p].
\eeq
Then 
\beq\label{eq:dimkerlb}
\dim\ker H_G\geq \mathcal Z(G';-p) d^{|G|}>0.
\eeq
\e{thm}

This can be considered a quantum analog of Shearer's theorem. We recall that the lower bound $\dim\ker H_G>0$ implies that $H_G$ is frustration-free. 


The main work of this section is then to verify condition \eqref{eq:Zpos} of Theorem \ref{thm:QSAT} in order to derive Lemma \ref{lm:FF}. We achieve this via the cluster expansion method, specifically the Koteck{\'y}-Preiss criterion.

\subsection{Cluster expansion bound}
\indent
Before we state the Koteck\'y-Preiss criterion, a bibliographical comment is in order.
The cluster expansion is used in this paper for a nearest-neighbor exclusion model on a graph.
In this simple setting, the cluster expansion should really be called a Mayer expansion, whose introduction dates back to \cite{mayer1937statistical,ursell1927evaluation}.
We will need an estimate for the error of the Mayer expansion, and such estimates can be found in the literature as far back as the early 1960's \cite{groeneveld1962two,ruelle1999statistical}.
Here, we have opted for the more general Koteck\'y-Preiss formulation \cite{kotecky1986cluster}, as it is directly applicable to discrete systems (graphs), whereas many of the older references are stated in the continuum (despite the fact that their proofs are easily generalizable to graphs).\\

We consider a graph $\mathcal G=(\mathcal K,\mathcal E)$ of bounded degree. We denote the maximal degree of $\mathcal G$ by $\Delta-1$. We let $\mathcal Z(\mathcal G;z)$ be as in Definition \ref{defn:Z}.

\begin{thm}[Koteck{\'y}-Preiss criterion \cite{kotecky1986cluster}]\label{theo:kotecky1986cluster}
  If there exists $a\geqslant 0$ such that, for all $\gamma\in \mathcal K$,
  \begin{equation}
    \sum_{\gamma':\ \gamma'\iota\gamma}e^{a}|z|\leqslant a
    \label{kotecky1986cluster_condition}
  \end{equation}
  then $\mathcal Z(\mathcal G;z)\neq0$, and there exists a function $\Phi^T:\mathcal B(\mathcal G)\to\mathbb C$ such that
  \begin{equation}
    \log\mathcal Z(\mathcal G;z)=\sum_{C\in\mathcal B(\mathcal G)}\Phi^T(C)
    \label{kotecky1986cluster_logZ}
  \end{equation}
  furthermore,
  \begin{equation}
    \sum_{C\in\mathcal B(\mathcal G):\ C\iota\gamma}|\Phi^T(C)|\leqslant a
    \label{kotecky1986cluster_bound}
  \end{equation}
  where $C\iota\gamma$ means that $\gamma'\iota\gamma$ for all $\gamma'\in C$.
\end{thm}

The following conclusion of this theorem is tailored to our purposes.

\begin{cor}\label{cor:CEbound}
For $0\leqslant p'\leqslant\frac{e^{-1}}\Delta$, we have
  \begin{equation}\label{eq:CEbound}
    \mathcal Z(\mathcal G;-p')\geqslant e^{-|K|e^1\Delta p'}.
  \end{equation}
\end{cor}

\begin{proof}[Proof of Corollary \ref{cor:CEbound}]
We introduce the specific free energy 
\begin{equation}
  f_{\mathcal G}(z)=\frac1{|\mathcal K|}\log\mathcal Z(\mathcal G;z),
  \label{freeen}
\end{equation}
with $|\mathcal K|$ the cardinality of $\mathcal K$. We claim that if
  \begin{equation}
    |z|\leqslant\frac{e^{-1}}\Delta
    \label{cdz}
  \end{equation}
then
\beq\label{eq:CEfbound}
    |f_{\mathcal G}(z)|
    \leqslant e^1\Delta|z|
  \eeq
 
  Suppose that \eqref{cdz} holds. We aim to apply Theorem\-~\ref{theo:kotecky1986cluster}.
  The condition\-~(\ref{kotecky1986cluster_condition}) is
  \begin{equation}
    |z|\leqslant\frac{ae^{-a}}\Delta
    .
  \end{equation}
  In order to choose $a$ as small as possible, we choose it to saturate the inequality
  \begin{equation}
    ae^{-a}=\Delta|z|
    .
    \label{a}
  \end{equation}
Indeed, since $a\mapsto ae^{-a}$ is an increasing function for $a\leqslant1$, condition \eqref{a} has a solution if and only if\-~(\ref{cdz}) holds.\\

  Furthermore, \eqref{kotecky1986cluster_logZ} implies that
  \begin{equation}
    |\log\mathcal Z(\mathcal G;z)|
    \leqslant\sum_{\gamma\in \mathcal K}
    \sum_{C\in\mathcal B(\mathcal G):\ C\iota\gamma}|\Phi^T(C)|
  \end{equation}
  and so, by\-~(\ref{kotecky1986cluster_bound}),
  \begin{equation}
    |f_{\mathcal G}(z)|
    \leqslant a
    \leqslant e^1\Delta|z|
    .
  \end{equation}
This proves \eqref{eq:CEfbound} and hence \eqref{eq:CEbound}.
\end{proof}

\subsection{Proof of frustration-freeness}

\be{proof}[Proof of Lemma \ref{lm:FF}]
We aim to apply Theorem \ref{thm:QSAT} and are thus led to consider $\mathcal Z(G',p')$. Elementary considerations show that the edge graph satisfies the maximal degree bound
$$
\De-1\leq 2(\de-1).
$$
Assumption \eqref{eq:mainass} ensures that
$$
p=\frac{r}{d^2}\leq \frac{e^{-1}}{2\de-1}\leq \frac{e^{-1}}{\De}.
$$
Hence, by Corollary \ref{cor:CEbound} with $\mathcal G=G'$, we have for all $p'\in [0,p]$ the estimate
$$
\mathcal Z(G';-p)\geq e^{-|E| e\De p}>0,
$$
where we used that the vertex set of the edge graph $G'$ is equal to the edge set $E$ of the original graph $G$. 

Hence, we can apply Theorem \ref{thm:QSAT} to conclude that 
$$
\dim\ker H_G\geq \mathcal Z(G';-p) d^{|G|}>0
$$
and so $H_G$ is frustration-free.
\e{proof}

\section{Bounds on the spectral gap}
\subsection{Proof of finite-size criterion for the local gap}
\begin{proof}[Proof of Proposition \ref{prop:knabelocal}]
Let $S\subset \Lam_L$ be an arbitrary subgraph. Without loss of generality, $S$ is connected. (Indeed, if $S$ is the disjoint union of $S_1$ and $S_2$, then $\gam(H_S)=\min\{\gam(H_{S_1}),\gam(H_{S_2})\}$.)  The following is an extension of Knabe's argument \cite{knabe1988energy} to any dimension without assuming isotropy. The underlying observation is that by restricting on the right-hand side of Proposition \eqref{eq:knabelocal} to $3$-vertex subgraphs, we can control the gap of any subgraph $S\subset \Lam_L$. 

By the spectral theorem and frustration-freeness, the claim is equivalent to the operator inequality
\beq\label{eq:HGsquaredrephrase}
H_S^2\geq (4D-2)\l(\gam_3-\frac{4D-3}{4D-2}\r) H_S.
\eeq
Since the $P^j_{x,x+e_j}$ are projections, we have
$$
H_S^2=H_S+Q_S+R_S
$$
where $Q_S$ contains the terms of the form $\{P^j_{x,x+e_j},P^k_{y,y+e_k}\}$ so that the edges $\{x,x+e_j\}$ and $\{y,y+e_j\}$ share exactly one vertex and $R_S $ contains the terms $\{P_{x,x+e_j},P_{y,y+e_k}\}$ so that the edges $\{x,x+e_j\}$ and $\{y,y+e_j\}$  do not share any vertex. Note that $R_S$ is  positive definite by commutativity. Hence,
\beq\label{eq:HGsquaredbound}
H_S^2\geq H_S+Q_S.
\eeq

We define the auxiliary operator 
$$
\begin{aligned}
\mathcal A_S=&\sum_{x\in S} \sum_{j=1}^D (H_{j,x})^2 \\
&+\sum_{x\in S} \l(\sum_{j,k=1}^D \l(\l(H^{++}_{jk,x}\r)^2
+\l(H^{-+}_{jk,x}\r)^2+\l(H^{--}_{jk,x}\r)^2\r)\r)
\end{aligned}
$$
where the subindex $x$ denotes a shift by $x$ compared to \eqref{eq:H3list}, e.g., $H^{-+}_{jk,x} =P^j_{x-e_j,x}+P^k_{x,x+e_k}$, etc.

Let $\de_x\leq 2D$ denote the degree of vertex $x\in S$ within the subgraph induced by $S$.
 On the one hand, we can expand the square and interchange the order of summation to find
$$
\mathcal A_S
=\sum_{x\in S}\sum_{j=1}^D (\de_x+\de_{x+e_j}-2) P^j_{x,x+e_j}
+Q_S.
$$
On the other hand, by translation-invariance, frustration-freeness, and the spectral theorem, we have the operator inequality
$$
\begin{aligned}
\mathcal A_S
\geq&\gam_3 \sum_{x\in S} \l(\sum_{j=1}^D H_{j,x} +\sum_{j,k=1}^D \l(H^{++}_{jk,x}
+H^{-+}_{jk,x}+H^{--}_{jk,x}\r)\r)\\
=&\gam_3\sum_{x\in S}\sum_{j=1}^D (\de_x+\de_{x+e_j}-2) P^j_{x,x+e_j}
\end{aligned}
$$
Combining these, we obtain
$$
\begin{aligned}
Q_S\geq &\l(\gam_3-1\r)  \sum_{x\in S}\sum_{j=1}^D (\de_x+\de_{x+e_j}-2) P^j_{x,x+e_j}\\
\geq& (4D-2)\l(\gam_3-1\r)  H_S,
\end{aligned}
$$
where the second step uses $\gam_3\leq 1$ and $1\leq \de_x\leq 2D$.
We now use \eqref{eq:HGsquaredbound} to conclude \eqref{eq:HGsquaredrephrase}.
\e{proof}

\subsection{Proof of the main result}\label{ssect:proofmain}
In this section, we prove the main result, Theorem \ref{thm:main}. Some details are parallel to \cite{lemm2019gaplessness} and will be skipped. An important difference here compared to \cite{lemm2019gaplessness} is the choice of ``good'' vectors $w_{ij}$ in \eqref{eq:wchoice}; see also Appendix A.

\be{proof}[Proof of Theorem \ref{thm:main}] By applying Lemma \ref{lm:FF} with $G=\Lam_L$, so $\de=2D$ and using \eqref{eq:mainass},we obtain that $H_{\Lam_L}$ is frustration-free. Hence we can apply Proposition \ref{prop:knabelocal} by which it suffices to show
\beq\label{eq:gam3claim}
\gam_3>\frac{4D-3}{4D-2}.
\eeq

We can lower bound $\gam_3$ by using the anticommutator bound from Lemma 6.3(ii) in \cite{fannes1992finitely}. 
Recall \eqref{eq:H3list} and introduce the set of pairs of projections appearing in it, i.e.,
$$
\begin{aligned}
\mathcal P=&\left\{
\{P^j_{-e_j,0},P^j_{0,e_j}\},
\{P^j_{0,e_j},P^k_{0,e_k}\},\r.\\
&\;\l.\{P^j_{-e_j,0},P^k_{0,e_k}\},
\{P^j_{-e_j,0},P^k_{-e_k,0}\}\right\}_{j,k=1,\ldots,D}
\end{aligned}
$$
Inspecting the proof of Theorem 5 in \cite{lemm2019gaplessness}, we obtain the gap bound
$$
\gam_3\geq 1- \max_{\{P,Q\}\in \mathcal P}\|P Q-P\wedge Q\|
$$
where $P\wedge Q$ is the orthogonal projection onto the subspace $\mathrm{ran}\, P\cap \mathrm{ran}\, Q$.

Hence, the claim \eqref{eq:gam3claim} reduces to
\beq\label{eq:gam3claim'}
\max_{\{P,Q\}\in \mathcal P}\|PQ-P\wedge Q\|
<\frac{1}{4D-2}.
\eeq

To prove \eqref{eq:gam3claim'}, we proceed via an auxiliary family of ``good'' projectors $\tilde P^1,\ldots,\tilde P^D$. (This generalizes an idea of \cite{lemm2019gaplessness}, though, the choice of good projectors is quite different here.) First, we set
\beq\label{eq:wchoice}
w_{ij}=\vert i\otimes j\rangle,\qquad \textnormal{for $i\leq d$ odd and $j\leq d$ even}.
\eeq
Note that the total number of these vectors is $\frac{d^2}{4}$ if $d$ is even and $\frac{d^2-1}{4}$ if $d$ is odd. Equivalently, the total number of these vectors is $\tilde d:=\lfloor \frac{d^2}{4}\rfloor$. Let 
$
\{v_1,\ldots,v_{\tilde d}\}
$
be an enumeration of the set $\{w_{i,j}\}$. Now we define the rank-$r$ good projectors
\beq\label{eq:tildePdefn}
\tilde P^j=\sum_{i=(j-1)r+1}^{jr} P_{v_i},\qquad j=1,\ldots,D,
\eeq
where $P_{v_i}$ projects onto $v_i$. Note that defining $\tilde P^D$ uses our assumption in \eqref{eq:mainass} that $Dr\leq \tilde d=\lfloor \frac{d^2}{4}\rfloor$.

It is then elementary to check from \eqref{eq:wchoice} the key property that
\beq\label{eq:tildePvanish}
\tilde P\tilde Q=\tilde P\wedge \tilde Q=0
\eeq
for all pairs $\{\tilde P,\tilde Q\}$ belonging to the set
$$
\begin{aligned}
\tilde{\mathcal{P}}=&\left\{
\{\tilde P^j_{-e_j,0},\tilde P^j_{0,e_j}\},
\{\tilde P^j_{0,e_j},\tilde P^k_{0,e_k}\},\r.\\
&\;\l.\{\tilde P^j_{-e_j,0},\tilde P^k_{0,e_k}\},
\{\tilde P^j_{-e_j,0},\tilde P^k_{-e_k,0}\}\right\}_{j,k=1,\ldots,D}
\end{aligned}
$$
Indeed, this is obvious for $j\neq k$ because the ranges of $\tilde P$ and $\tilde Q$ are orthogonal, while for $\{\tilde P,\tilde Q\}=\{\tilde P^j_{-e_j,0},\tilde P^j_{0,e_j}\}$ this relies on the parity structure in \eqref{eq:wchoice} and the fact that even indices are paired with odd ones at the center site $0$.

Now we proceed with a perturbative argument. Specifically, we apply Lemma 7 in \cite{lemm2019gaplessness}, which follows from some geometrical estimates for spherical caps, to the particular collection of vectors $\{v_1,\ldots,v_{\tilde d}\}$. This yields, for any $\eps>0$, vectors $\{\phi_1,\ldots,\phi_{\tilde d}\}$ such that
\beq\label{eq:lemma7}
\mathbb P\l(\sup_{1\leq i\leq \tilde d} \|v_i-\phi_i\|<\eps \r)\geq p_\eps>0.
\eeq
From now on, we assume that, for some $\eps>0$ to be determined, we have taken $\{\phi_1,\ldots,\phi_{\tilde d}\}$ such that this event occurs, i.e., $\|v_i-\phi_i\|<\eps$ holds for all $1\leq i\leq \tilde d$.  

We recall a classical theorem of von Neumann \cite{neumann1950functional} which says that for any pair of projections $P,Q$,
$$
P\wedge Q=s-\lim_{n\to\infty} (PQ)^n,
$$
and consequently $\|P\wedge Q\|\leq \liminf_{n\to\infty} \|PQ\|^n$, leading us to
\beq\label{eq:PQ}
\|PQ-P\wedge Q\|
\leq \|PQ\| 
+ \liminf_{n\to\infty} \|PQ\|^n.
\eeq
Given $\{P,Q\}\in\mathcal P$, we let $\{\tilde P,\tilde Q\}$ be the corresponding pair in $\mathcal P$. Then we estimate $\|PQ\|= \|PQ-\tilde P\tilde Q\| $ by using $\|v_i-\phi_i\|<\eps$ and following the same steps as in \cite{lemm2019gaplessness} to conclude that
\beq
\|PQ\|\leq 4r\eps.
\eeq
For $\eps<\frac{1}{4r}$, we obtain $\liminf_{n\to\infty} \|PQ\|^n=0$. Finally, we choose $\eps>0$ small enough such that $4r\eps<\frac{1}{4D-2}$. Then \eqref{eq:PQ} implies \eqref{eq:gam3claim'} and hence Theorem \ref{thm:main}.
\e{proof}

\section{Extension to other directed graphs}
\label{sect:gen}
In this short section, we explain how the proof method can be generalized to other directed graphs. After some general remarks, we focus for definiteness on the specific example of the \textit{honeycomb lattice} which is commonplace in applications; see, e.g., \cite{lemm2020existence,pomata2020demonstrating}.

\subsection{General discussion}

Recall Definition \eqref{eq:HLdefn} of the Hamiltonian and observe the following two features: (i) the Hamiltonian is defined on oriented edges (the orientation matters because $P^j_{x,x+e_j}\neq P^j_{x+e_j,x}$) and (ii) the symmetry structure of the Euclidean lattice enters because it allows to set interactions on collinear edges equal. Point (ii) is no longer available for a general directed graph which does not necessarily have obvious ``partner'' directions such as the collinear ones in the Euclidean case, though, other lattice graphs come with their own symmetries. 

At any rate, also without symmetries, we can define a random Hamiltonian on a general $k$-regular directed graph by independently sampling the prototypical $k$ local interactions $P^1,\ldots,P^k$ at any vertex at random. That is, all local interactions are sampled independently, which is different from the construction in the Euclidean case. Once a random translation-invariant Hamiltonian is defined in this way, the argument can be adapted. This is possible because we are able to ensure frustration-freeness for general graphs via Lemma \ref{lm:FF} thanks to the cluster expansion. It is notationally slightly cumbersome to formulate an explicit construction of the random translation-invariant Hamiltonian on general $k$-regular directed graphs after the local interaction prototypes $P^1,\ldots,P^k$ have been sampled. Hence, we find it most instructive to consider an explicit example and so we focus on the honeycomb lattice in the following.

\subsection{The honeycomb lattice}
Let $\mathbb H_L=(K,E)$ denote the honeycomb lattice wrapped on an $L\times L$ torus. We place a qudit $\C^d$ at each site of $\mathbb H_L$, so we consider the Hilbert space $\bigotimes_{x\in \mathbb H_L} \C^d$. Fix any orientation of the edges. Let $r\geq 1$. As before, let $P^1,P^2,P^3:\C^d\otimes \C^d\to\C^d\otimes \C^d$ denote three i.i.d.\ random rank $r$ projections with the distribution of each $P^j$ given as the projection onto the first $r$ columns of a Haar-random orthogonal matrix. The random translation-invariant Hamiltonian on the honeycomb lattice is given by
$$
H_{\mathbb H_L}=\sum_{\vec{e}\in E} P^{t(\vec{e})}_{\vec{e}}
$$
Here, for any directed edge $\vec{e}=(x,y)$, we set $P^{j}_{\vec{e}}=P_{x,y}\otimes \mathrm{Id}_{\mathbb H_L\setminus \{x,y\}}$ and we introduced the type $t(\vec{e})$,
$$
t(\vec{e})=
\be{cases}
1,\qquad \mathrm{span}\{y-x\}\textnormal{ parallel to $(1,0)$-axis},\\
2,\qquad \mathrm{span}\{y-x\}\textnormal{ at angle $\frac{\pi}{3}$ to $(1,0)$-axis},\\
3,\qquad \mathrm{span}\{y-x\}\textnormal{ at angle $\frac{2\pi}{3}$ to $(1,0)$-axis}.
\e{cases}
$$

\be{thm}[Honeycomb lattice]
Assume that 
$$
r\leq \min\l\{\frac{\lfloor\frac{d^2}{4}\rfloor}{3},\frac{e^{-1}}{5}d^2\r\}
$$
Then, there exist $c>0$ and a positive probability $p>0$ such that
$$
\mathbb P(\gam_{\mathrm{loc},\Lam_L}\geq c)\geq p>0.
$$
\e{thm}

\be{proof}
Since $r\leq \frac{e^{-1}}{5}d^2$, we can apply Lemma \ref{lm:FF} to conclude that $\mathbb H_L$ is frustration-free. The proof of the finite-size criterion in Proposition \ref{prop:knabelocal} generalizes straightforwardly to any frustration-free Hamiltonian on a bounded degree graph; in that case, $\gam_3$ is the minimal gap that occurs when restricting the corresponding Hamiltonian to $3$ connected sites. Using $r\leq \frac{\lfloor\frac{d^2}{4}\rfloor}{3}$, the good projections $\tilde P^1,\tilde P^2,\tilde P^3$ can be defined exactly as in \eqref{eq:tildePdefn} in terms of the same good vectors \eqref{eq:wchoice}. The same argument as before implies \eqref{eq:tildePvanish}. Then we use \cite[Lemma 7]{lemm2019gaplessness} as the basis of the perturbative argument. The details are analogous to the Euclidean lattice case and are hence ommitted.
\e{proof}

\section*{Acknowledgments}
We are grateful to Anurag Anshu and Zeph Landau for kindly informing us of their observation \cite{AL} that the main result of \cite{anshu2022area} extends to degenerate ground spaces in the form of Theorem \ref{thm:ALcrit}. We acknowledge support by Open Access Publishing Fund of University of Tübingen.

\bibliographystyle{quantum}
\bibliography{gaps}

\appendix

\section{Improvement of the main result in \cite{lemm2019gaplessness}}\label{app:improvementL}
 \be{proof}[Proof of Theorem \ref{thm:improvementL}]
Let $r\leq \max\{\tfrac{d^2}{4},d-1\}$. Without loss of generality, we assume that $r=\lfloor\tfrac{d^2}{4}\rfloor\geq d$. We take for granted the entire analysis in \cite{lemm2019gaplessness} and only mention where the critical change occurs. We replace the vectors
$$
v_i=\vert1\otimes (i+1)\rangle, \qquad \textnormal{for }1\leq i\leq r
$$
(of which there exist at most $d-1$, hence the restriction), by the family
$$
w_{ij}=\vert i\otimes j\rangle,\qquad \textnormal{for $i\leq d$ odd and $j\leq d$ even}.
$$
The key property of these vectors is that thanks to the parity structure we have
$$
\tilde P_{1,2}\tilde P_{2,3}=0,\qquad \textnormal{where } \tilde P=\sum_{\substack{i\leq d \textnormal{ odd}\\ j\leq d \textnormal{ even}}} w_{i,j}.
$$
By orthogonality, the rank of $\tilde P$ is given by
$$
\mathrm{rank}\, \tilde P
=\be{cases}
\frac{d^2}{4},\qquad &\textnormal{ if $d$ is even},\\
\frac{d^2-1}{4},\qquad &\textnormal{ if $d$ is odd}.
\e{cases}
$$
or equivalently $\mathrm{rank}\, \tilde P=\lfloor\frac{d^2}{4}\rfloor$ which equals $r$ by assumption.

The rest of the argument goes exactly as in \cite{lemm2019gaplessness} \textit{mutatis mutandis}. We omit the details.
 \e{proof}

%

\end{document}